\documentclass[12pt]{article}
\usepackage{graphicx,color}
\usepackage{lscape}
\usepackage{epsfig}
\usepackage{subfigure}
\usepackage{hhline,amssymb,epsfig,amsmath,array,amsthm}
\begin{document}

\newcommand{\ga}{{\alpha}}
\newcommand{\gb}{{\beta}}
\newcommand{\gc}{{\chi}}
\newcommand{\gd}{{\delta}}
\newcommand{\gep}{{\epsilon}}
\newcommand{\gvare}{\varepsilon}
\newcommand{\gga}{{\gamma}}
\newcommand{\gk}{{\kappa}}
\newcommand{\gl}{{\lambda}}
\newcommand{\gm}{{\mu}}
\newcommand{\gn}{{\nu}}
\newcommand{\go}{{\omega}}
\newcommand{\gp}{{\pi}}
\newcommand{\gph}{{\phi}}
\newcommand{\gvp}{{\varphi}}
\newcommand{\gps}{{\psi}}
\newcommand{\gth}{{\theta}}
\newcommand{\gr}{{\rho}}
\newcommand{\gs}{{\sigma}}
\newcommand{\gt}{{\tau}}
\newcommand{\gx}{{\xi}}
\newcommand{\gz}{{\zeta}}
\newcommand{\gE}{{\Upsilon}}
\newcommand{\gGa}{{\Gamma}}
\newcommand{\gL}{{\Lambda}}
\newcommand{\gO}{{\Omega}}
\newcommand{\gP}{{\Pi}}
\newcommand{\gPh}{{\Phi}}
\newcommand{\gPs}{{\Psi}}
\newcommand{\gTh}{{\Theta}}
\newcommand{\gS}{{\Sigma}}
\newcommand{\gX}{{\Xi}}
\newcommand{\gU}{{\Upsilon}}

\newcommand{\fa}{{\bf a}}
\newcommand{\fb}{{\bf b}}
\newcommand{\fc}{{\bf c}}
\newcommand{\fd}{{\bf d}}
\newcommand{\fe}{{\bf e}}
\newcommand{\ff}{{\bf f}}
\newcommand{\fk}{{\bf k}}
\newcommand{\fm}{{\bf m}}
\newcommand{\fn}{{\bf n}}
\newcommand{\fo}{{\bf o}}
\newcommand{\fp}{{\bf p}}
\newcommand{\fq}{{\bf q}}
\newcommand{\fr}{{\bf r}}
\newcommand{\fs}{{\bf s}}
\newcommand{\ft}{{\bf t}}
\newcommand{\fu}{{\bf u}}
\newcommand{\fv}{{\bf v}}
\newcommand{\fw}{{\bf w}}
\newcommand{\fx}{{\bf x}}
\newcommand{\fy}{{\bf y}}
\newcommand{\fz}{{\bf z}}

\newcommand{\fA}{{\bf A}}
\newcommand{\fB}{{\bf B}}
\newcommand{\fC}{{\bf C}}
\newcommand{\fD}{{\bf D}}
\newcommand{\fE}{{\bf E}}
\newcommand{\fF}{{\bf F}}
\newcommand{\fG}{{\bf G}}
\newcommand{\fI}{{\bf I}}
\newcommand{\fJ}{{\bf J}}
\newcommand{\fK}{{\bf K}}
\newcommand{\fL}{{\bf L}}
\newcommand{\fM}{{\bf M}}
\newcommand{\fN}{{\bf N}}
\newcommand{\fO}{{\bf O}}
\newcommand{\fP}{{\bf P}}
\newcommand{\fQ}{{\bf Q}}
\newcommand{\fR}{{\bf R}}
\newcommand{\fS}{{\bf S}}
\newcommand{\fT}{{\bf T}}
\newcommand{\fU}{{\bf U}}
\newcommand{\fV}{{\bf V}}
\newcommand{\fW}{{\bf W}}
\newcommand{\fX}{{\bf X}}
\newcommand{\fY}{{\bf Y}}
\newcommand{\fZ}{{\bf Z}}

\newcommand{\cB}{{\cal B}}
\newcommand{\cE}{{\cal E}}
\newcommand{\cG}{{\cal G}}
\newcommand{\cS}{{\cal S}}
\newcommand{\cF}{{\cal F}}
\newcommand{\cT}{{\cal T}}
\newcommand{\cD}{{\cal D}}
\newcommand{\cI}{{\cal I}}
\newcommand{\cK}{{\cal K}}
\newcommand{\cL}{{\cal L}}
\newcommand{\cM}{{\cal M}}
\newcommand{\cN}{{\cal N}}
\newcommand{\cU}{{\cal U}}
\newcommand{\cW}{{\cal W}}
\newcommand{\cX}{{\cal X}}
\newcommand{\cY}{{\cal Y}}
\newcommand{\cZ}{{\cal Z}}

 \newcommand{\fga}{\mbox{\boldmath $\alpha$}}
  \newcommand{\fgb}{\mbox{\boldmath $\beta$}}
  \newcommand{\fgd}{\mbox{\boldmath $\delta$}}
 \newcommand{\fgg}{\mbox{\boldmath $\gamma$}}
  \newcommand{\fgl}{\mbox{\boldmath $\lambda$}}
  \newcommand{\fgm}{\mbox{\boldmath $\mu$}}
 \newcommand{\fgs}{\mbox{\boldmath $\sigma$}}
 \newcommand{\fgth}{\mbox{\boldmath $\theta$}}
  \newcommand{\fgt}{\mbox{\boldmath $\tau$}}
  \newcommand{\fgz}{\mbox{\boldmath $\zeta$}}
  \newcommand{\fgps}{\mbox{\boldmath $\psi$}}
  \newcommand{\fge}{\mbox{\boldmath $\eta$}}
  \newcommand{\fgga}{\mbox{\boldmath $\gamma$}}
  \newcommand{\fgv}{\mbox{\boldmath $\varphi$}}
  \newcommand{\fgp}{\mbox{\boldmath $\pi$}}
  \newcommand{\fvare}{\mbox{\boldmath $\varepsilon$}}
  \newcommand{\fgD}{{\bf \Delta}}
  \newcommand{\fgG}{{\bf \Gamma}}
  \newcommand{\fgL}{{\bf \Lambda}}
  \newcommand{\fgPh}{{\bf \Phi}}
  \newcommand{\fgPi}{{\bf \Pi}}
  \newcommand{\fgPs}{{\bf \Psi}}
  \newcommand{\fgS}{{\bf \Sigma}}
  \newcommand{\fgTh}{{\bf \Theta}}
  \newcommand{\fgO}{{\bf \Omega}}
  \newcommand{\fgX}{{\bf \Xi}}
  \newcommand{\fgU}{{\bf \Upsilon}}

\newcommand{\Mean}{{\mbox{E}}}
\newcommand{\Cov}{{\mbox{cov}}}
\newcommand{\Var}{{\mbox{var}}}
\newcommand{\Corr}{{\mbox{corr}}}
\newcommand{\diag}{{\mbox{diag}}}
\newcommand{\prob}{{\mbox{Pr}}}

\newcommand{\Nul}{{\bf 0}}
\newcommand{\One}{{\bf 1}}
\newcommand{\Bd}{\B_{\bullet}}
\newcommand{\Xd}{\X_{\bullet}}
\newcommand{\Zd}{\Z_{\bullet}}
\newcommand{\Nor}{N}
\newcommand{\f}{\textbf}
\newcommand{\Real}{{\cal R}}
\newcommand{\Natural}{{\cal N}}
\newtheorem{thm}{Theorem}%[section]
\newtheorem{cor}[thm]{Corollary}
\newtheorem{lem}[thm]{Lemma}
\newtheorem{prop}[thm]{Proposition}
\newtheorem{ax}{Axiom}
\theoremstyle{definition}
\newtheorem{defn}{Definition}[section]
\theoremstyle{remark}
\newtheorem{rem}{Remark}[section]
\newtheorem*{notation}{Notation}
%\numberwithin{equation}{section}
%\newcommand{\thmref}[1]{Theorem~\ref{#1}}
\newcommand{\secref}[1]{\S\ref{#1}}
\newcommand{\bysame}{\mbox{\rule{3em}{.4pt}}\,}
\newcommand{\A}{\mathcal{A}}
\newcommand{\B}{\mathcal{B}}
\newcommand{\st}{\sigma}
\newcommand{\XcY}{{(X,Y)}}
\newcommand{\SX}{{S_X}}
\newcommand{\SY}{{S_Y}}
\newcommand{\SXY}{{S_{X,Y}}}
\newcommand{\SXgYy}{{S_{X|Y}(y)}}
\newcommand{\Cw}[1]{{\hat C_#1(X|Y)}}
\newcommand{\G}{{G(X|Y)}}
\newcommand{\PY}{{P_{\mathcal{Y}}}}
\newcommand{\X}{\mathcal{X}}
\newcommand{\wt}{\widetilde}
\newcommand{\wh}{\widehat}

\newcommand{\bY}{\mathbf{Y}}
\newcommand{\bV}{\mathbf{V}}
\newcommand{\bG}{\mathbf{G}}
\newcommand{\bU}{\mathbf{U}}
\newcommand{\bM}{\mathbf{M}}
\newcommand{\bH}{\mathbf{H}}
\newcommand{\bD}{\mathbf{D}}
\newcommand{\bI}{\mathbf{I}}
\newcommand{\by}{\mathbf{y}}
\newcommand{\bx}{\mathbf{x}}
\newcommand{\bv}{\mathbf{v}}
\newcommand{\bz}{\mathbf{z}}
\newcommand{\btheta}{\boldsymbol{\theta}}
\newcommand{\bpsi}{\boldsymbol{\psi}}
\newcommand{\bX}{\mathbf{X}}
\newcommand{\bepsilon}{\boldsymbol{\epsilon}}
\newcommand{\bbeta}{\boldsymbol{\beta}}

\begin{center}
{\LARGE\bf Composite likelihood estimation of sparse Gaussian graphical models with symmetry }
\medskip
\medskip
\medskip

% last updated 3/12/09

B{\scriptsize Y} XIN GAO\\
{\it Department of Mathematics and Statistics, York University, Toronto, Onatrio \\ Canada M3J 1P3} \\
xingao@mathstat.yorku.ca\\
{\scriptsize AND} HELENE MASSAM\\
{\it Department of Mathematics and Statistics, York University, Toronto, Onatrio \\ Canada M3J 1P3} \\
massamh@mathstat.yorku.ca
\end{center}

\newpage

\centerline{ABSTRACT}

\noindent

In this article, we discuss the composite likelihood estimation of sparse Gaussian graphical models. When there are symmetry constraints on the concentration matrix or partial correlation matrix, the likelihood estimation can be computational intensive. The composite likelihood offers an alternative formulation of the objective function  and yields consistent estimators. When a sparse model is considered, the penalized composite likelihood estimation can yield estimates satisfying both the symmetry and sparsity constraints and possess ORACLE property. Application of the proposed method is demonstrated through simulation studies and a network analysis of a biological data set.


\begin{center}
1. I{\scriptsize NTRODUCTION}
\end{center}

A multivariate Gaussian graphical model is also known as 
covariance selection model. The conditional independence
relationships between the random variables are equivalent to 
specified zeros among the inverse covariance matrix. More exactly,
let $X = (X^{(1)},. . . ,X^{(p)})$ be a $p$-dimensional random
vector following a multivariate normal distribution
$N_p(\mu,\Sigma),$ with $\mu$ denoting the unknown mean and
$\Sigma$ denoting the nonsingular covariance matrix. Denote the
inverse covariance matrix as $\Sigma^{-1}=C=(C_{ij})_{1\leq
i,j\leq p}.$ Zero entries $C_{ij}$ in the inverse covariance
matrix indicate conditional independence between the random variables $X^{(i)}$ and $X^{(j)}$ given all other variables
(Dempster (1972), Whittaker (1990), Lauritzen (1996)). The Gaussian
random vector $X$ can be represented by an undirected graph
$G=(V,E),$ where $V$ contains $p$ vertices corresponding to the
$p$ coordinates and the edges $E=(e_{ij})_{1\leq i < j \leq p}$
represent the conditional dependency relationships between
variables $X^{(i)}$ and $X^{(j)}.$ It is of interest to identify
the correct set of edges, and estimate the parameters in the
inverse covariance matrix simultaneously.

To address this problem, many methods have been developed. In general, there are no zero entries in the maximum
likelihood estimate, which results in a full graphical structure.
Dempster (1972) and Edwards (2000) proposed to use penalized
likelihood with the $L_0$-type penalty
$p_{\lambda}(|c_{ij}|)_{i\neq j}=\lambda I(|c_{ij}|\neq 0)$, where
$I(.)$ is the indicator function. Since the $L_0$ penalty is
discontinuous, the resulting penalized likelihood estimator is
unstable. Another approach is stepwise forward selection or backward
elimination of the edges. However, this ignores the
stochastic errors inherited in the multiple stages of the
procedure (Edwards (2000)) and the statistical properties of
the method are hard to comprehend.  Meinshausen and
B\"{u}hlmann (2006) proposed a computationally attractive method
for covariance selection; it performs the
neighborhood selection for each node and combines the results to
learn the overall graphical structure. Yuan and Lin (2007) proposed
penalized likelihood methods for estimating the concentration
matrix with the $L_1$ penalty (LASSO) (Tibshirani (1996)). Banerjee, Ghaoui, and D'aspremont (2007) proposed a block-wise updating algorithm for the estimation of the inverse covariance
matrix. Further in this line,
Friedman, Hastie, and Tibshirani (2008) proposed the graphical
LASSO algorithm to estimate the sparse inverse covariance matrix
using the LASSO penalty through a coordinate-wise updating scheme. Fan, Feng, and Wu (2009) proposed to estimate
the inverse covariance matrix using the adaptive LASSO and the Smoothly
Clipped Absolute Deviation (SCAD) penalty to attenuate the bias
problem. Friedman, Hastie and Tibshirani (2012) proposed to use composite likelihood based on conditional likelihood to estimate sparse graphical models.

In real applications, there often exists symmetry constraints on the underlying Gaussian graphical model. For example, genes belong to the same functional or structure group may behave in a similar manner and thus share similar network properties.  In the analysis of high-dimensional data, clustering algorithm is often performed to reduce the dimensionality of the data. Variates in the same cluster exhibit similar patterns.  This may result in restrictions on the graphical gaussian models: equality among sepcified elements of the concentration matrix or equality emong specific partial variances and correlations. Adding symmetry to the graphical model reduces the number of parameters. When both sparsity and symmetry exisits, the likelihood estimation becomes computationally challenging. 

Hojsgaard and Lauritzen (2009) introduced new types of Guassian models with symmetry constraints. When the restriction is imposed on the inverse convariance matrix, the model is referred as RCON model. When the restriction is imposed on the partial correlation matrix, the model is referred as RCOR model. Likelihood estimation on both models can be obtained through Newton iteration or partial maximization. However, the algorithm involves the inversion of concentration matrix in the interation steps, which can be computationally costly in the analysis of large matrices. When sparsity constrainst is imposed on the RCON and RCOR model, the likelihood is added extra penalty terms on the sizes of the edges. Solving the penalized likelihood with both sparsity and symmetry constraint is a challenge. In this article, we investigate the alternative way of formulating the likelihood. We propose to use composite likelihood as our objective function and maximize the penalized composite likelihood to obtain the sparse RCON and RCOR model. The algorithm is designed based on co-ordinate descent and soft thresholding rules. The algorithm is computationally convenient and it avoids any operations of large matrix inverison.

The rest of the article is organized as follows. In Section 2.1 we
formulate the penalized likelihood function for the RCON and RCOR modle
matrix. In Sections 2.2 and 2.3, we present the coordinate descent algorithm and soft thresholding rule. In Section 3, we
investigate the asymptotic behavior of the estimate and establish the ORACLE property of the estimate. In Section 4, simulation studies are presented
to demonstrate the empirical performance of the estimate in terms of estimation and model selection. In Section 5, we applied our method to a clustered microarray data set to estimate the networks between the clustered genes and also compare the networks under different treatment settings.

\begin{center}
2. M{\scriptsize ETHOD}
\end{center}
\begin{center}
2.1  C{\scriptsize OMPOSITE LIKELIHOOD}
\end{center}
The estimation of Gaussian graphical model has been mainly based on likelihood method. An alternative method of estimation based on composite likelihood has
drawn much attention in recent years.  It has been demonstrated to possess good
theoretical properties, such as consistency for the parameter
estimation, and can be utilized to establish hypothesis testing
procedures.  Let $x=(x_1,\ldots,x_n)^T$ be the vector of $n$
variables observed from a single observation. Let $\{f(x;\phi),x\in
\mathcal{X},\phi \in \Psi\}$ be a class of parametric models,
with $\mathcal{X}\subseteq \mathcal{R}^n,$ $\Psi \subseteq
\mathcal{R}^q,$ $n\geq 1,$ and $q \geq 1.$ For a subset of
$\{1,\ldots,n\}$, say $a$, $x_a$ denotes a subvector of $x$
with components indexed by the elements in set $a;$ for instance,
given a set $a=\{1,2\}$, $x_{a} = (x_1,x_2)^T$. Let $\phi
=(\theta,\eta)$, where $\theta \in \Theta \subseteq \mathcal{R}^p$, $p
\leq q$, is the parameter of interest, and $\eta$ is the nuisance
parameter. According to Lindsay (1988), the CL of a single
vector-valued observation is $ L_c(\theta;x)=\prod_{a \in
A}L_a(\theta;x_a)^{w_a}, $ where $A$ is a collection of index
subsets called the composite sets,
$L_a(\theta;x_a)=f_a(x_a;\theta_a),$ and $\{w_a,a\in A\}$ is a
set of positive weights. Here $f_a$ denotes all the different
marginal densities and $\theta_a$ indicates the parameters that
are identifiable in the marginal density $f_a.$  

As the composite score function is a linear combination of several
valid likelihood score functions, it is unbiased under the usual
regularity conditions.  Therefore, even though the composite
likelihood is not a real likelihood, the maximum composite
likelihood estimate is still consistent for the true parameter.
The asymptotic covariance matrix of the maximum composite
likelihood estimator takes the form of the inverse of the Godambe
information:$
H(\theta)^T J(\theta)^{-1}H(\theta),
$
where $H(\theta)=E\{- \sum_{a\in A} \partial^2 \log
f(x_a;\theta)/\partial \theta \partial \theta^T\}$ and
$J(\theta)=\text{var}\{\sum_{a\in A} \partial \log f(x_a;
\theta)/\partial \theta\}$ are the sensitivity matrix and the
variability matrix, respectively. Readers are referred
to Cox and Reid (2004) and Varin (2008) for a more detailed
discussion on the asymptotic behavior of the maximum composite
likelihood estimator.

\begin{center}
2.1  C{\scriptsize OMPOSITE LIKELIHOOD ESTIMATION OF RCON MODEL}
\end{center}
Let data $X$ consist of $n$ replications of a multivariate random vector of size $p$: $X=(X_1,X_2,\dots,X_n)^T,$ with $X_i=(X_{i1},X_{i2},\dots,X_{ip})^T$ following a $N_p(\mu,\Sigma)$ distribution. For simplicity of exposition, we assume throughout that $\mu=0.$  We let $\theta=\Sigma^{-1}$ denote the inverse covariance, also known as the concentration matrix with elements $(\theta_{ij})$, $1\leq i,j,\leq p.$ The partial correlation between $X_{ij}$ and $X_{ik}$ given all other variables is then
$$
\rho_{jk}=-\theta_{jk}/\sqrt{\theta_{jj}\theta_{kk}}.
$$
It can be shown than $\theta_{jk}=0$ if and only if $X_{ij}$ and $X_{ik}$ are conditionally independent given all other variables.

There are different symmetry restrictions on cencentrations first introduced by Hojsgaard and Lauritzen (2009). An $RCON(\mathcal{V},\mathcal{E})$ model with vertex coloring $\mathcal{V}$ and edge coloring $\mathcal{E}$ is obtained by restricting the elements of the concentration matrix $\theta$ as follows: 1) Diagonal elements of the concentration matrix $\theta$  corresponding to vertices in the same vertex colour class must be identical. 2) Off diagonal entries of $\theta$ corresponding to edges in the same edge colour class must be identical. Let $\mathcal{V}=\{V_1,\dots,V_k\}$, where $V_1,\dots,V_k$ is a partition of $\{1,\dots,p\}$ vertex class. Let $\mathcal{E}=\{E_1,\dots,E_l\}$, where $E_1,\dots,E_l$ is a partition of $\{(i,j),1\leq i<j \leq p\}$ edge class. This implies given an edge color class, for all edges $(i,j)\in E_s,$ $\theta_{ij}$ are all equal and hence denoted as $\theta_{E_s}.$ This also implies given a vertex color class, for all vertices $(i)\in V_m,$ $\theta_{ii}$ are all equal and hence denoted as $\theta_{V_m},$ $\sigma^{ii}$ are all equal and hence denoted as $\sigma_{V_m},$

Following the approach of Friedman, Hastie and Tibshirani (2012), we formulate composite conditional likelihood to estimate sparse graphical model under symmetry constraints.  The conditional distribution of $x_{ij}|x_{-ij}=N(\sum_{k \neq j } x_{ik}\beta_{kj}, \sigma^{jj}), $ where $x_{-ij}=(x_{i1},x_{i2},\dots,x_{i,j-1},x_{j+1},\dots,x_{ip}),$ $\beta_{kj}=-\theta_{kj}/\theta_{jj},$ and $\sigma^{jj}=1/\theta_{jj}.$ The negative composite log-likelihood can be formulated as
$$
\ell_c(\theta)=\frac 1 2 \sum_{j=1}^p(N\log \sigma^{jj}+\frac 1 {\sigma^{jj}}||X_j-XB_{j}||_2^2),
$$
where $B_j$ is a $p-$vector with elements $\beta_{ij},$ except a zero at the j$th$ position, and $B=(B_1,B_2,\dots,B_p).$ We propose to estimate the sparse RCON model by minimizing the following penalized composite loglikelihood $Q(\theta)$:
$$
\min_{\theta_{E_s},1\leq s \leq l,\theta_{V_m},1\leq m \leq k}\ell_c(\theta)+n\lambda\sum_{s}|\theta_{E_s}|
.$$

We employ coordinate-descent algorithm by solving the penalized minimization one coordinate at a time. It can be shown that the negative expected Hessian matrix of $\ell_c(\theta)$ is positive definite because it is the sum of expected negative Hessian matrices of all conditional likelihoods: 
\begin{align}
\begin{split}
&E(\frac{-\partial^2 \ell_c(\theta)}{\partial \theta^2}) =\sum_{i=1}^n \sum_{j=1}^p E(\frac{\partial^2 l(x_{ij}|x_{-ij})}{\partial \theta^2})\\
=& \sum_{i=1}^n \sum_{j=1}^p E(E(\frac{\partial^2 l(x_{ij}|x_{-ij})}{\partial \theta^2}|x_{-ij}))=\sum_{i=1}^n \sum_{j=1}^p E(\text{var}(\frac{\partial l(x_{ij}|x_{-ij})}{\partial \theta}|x_{-ij})).
\end{split}
\end{align}
Each $\text{var}(\frac{\partial l(x_j|x_{-j})}{\partial \theta}|x_{-j})$ is positive definite and integrals preserve positive definiteness, therefore $E(\frac{\partial^2 \ell_c(\theta)}{\partial \theta^2})
$ is positive definite. Thus, when $n$ is sufficiently larege, the objective function $Q(\theta)$  is locally convex at $\theta_0$. If the interation steps of the algorithm hits this neighborhood, the algorithm will converge to $\theta_0$.

The co-ordinate descent algorithm proceeds by updating each parameter of the objective function one at a time. The first derivative of the objective function with respect to the edge class parameter is as follows. The technical derivation is in the Appendix.
\begin{align}
\begin{split}
\frac{\partial Q(\theta)}{\partial \theta_{E_s}} 
=&( \sum_{j=1}^p\sum_{i;(i,j)\in E_s}\sum_{l;(l,j)\in E_s}\sigma^{jj}X_i^T X_l)\theta_{E_s}+\\
& \biggl(\sum_{j=1}^p X_j^T(\sum_{i;(i,j)\in E_s} X_i)+\sigma^{jj} \sum_{i;(i,j)\in E_s}\sum_{l;(l,j)\in E_s^c}X_i^T X_l \theta_{lj}\biggr)+n\text{sgn}(\theta_{E_s}),
\end{split}
\end{align}
where $E_s^c=\{(i,j)|i\neq j \,\text{and}\, (i,j)\notin E_s\}.$ Therefore the update for $\theta_{E_s}$ is
$$
\hat{\theta}_{E_s}=\frac{S(-(\sum_{j=1}^p X_j^T(\sum_{i;(i,j)\in E_s} X_i)+\sigma^{jj} \sum_{i;(i,j)\in E_s}\sum_{l;(l,j)\in E_s^c}X_i^T X_l \theta_{lj})/n,\lambda)}{( \sum_{j=1}^p\sum_{i;(i,j)\in E_s}\sum_{l;(l,j)\in E_s}\sigma^{jj}X_i^T X_l)/n},
$$
where $S(z,\lambda)=\text{sign}(z)(|z|-\lambda)_{+}$ is the soft-thresholding operator. Let $C=\frac 1 n X^T X$ denote the sample covariance matrix. Given the color edge group $E_s,$ we construct the edge adjancency matrix $T^{E_s},$ with $T^{E_s}_{ij}=1,$ if $(i,j)\in E_s,$ and $T^{E_s}_{ij}=0$ otherwise. We can simplify the updating expression as follows: 
$$
\hat{\theta}_{E_s}=\frac{S(- \text{tr}(T^{E_s} C)+\text{tr}(T^{E_s}(T^{E_s^c}\odot B) C),\lambda)}{\text{tr}(T^{E_s}(T^{E_s}\sigma) C)},
$$
where $\odot$ denotes the componentwise product,  and $\sigma$ denotes a $p\times p$ matrix of $\text{diag}(\sigma^{jj}).$ 

For notational convenience, let $\tilde{\theta}$ denote a $p\times p$ matrix with diagonal elements equal to zero and off-diagonal elements equal to that of $\theta.$ The first derivative of $Q(\theta)$ with respect to the vertex class is as follows:
\begin{align}
\begin{split}
&\frac{\partial Q(\theta)}{\partial \sigma_{V_m}} \\
=& \frac n 2\sum_{j \in V_m}( \frac 1 {\sigma^{jj}}-\frac{C_{jj}}{(\sigma^{jj})^2}+q_j),
\end{split}
\end{align}
where $q_j=\sum_{l=1}^p \sum_{l'=1}^p C_{ll'} \tilde{\theta}_{lj}\tilde{\theta}_{l'j}.$
Therefore the solution of
$$\hat{\sigma}_{V_m}=\frac{-|V_m|+\sqrt{|V_m|^2+4(\sum_{j\in V_m}q_j)(\sum_{j\in V_m}C_{jj})}}{2\sum_{j \in V_m} q_j},
$$
where $|V_m|$ denotes the cardinality of $V_m.$ Let diagonal matrix $T^{V_m}$ denote the generator for the vertex color class, with $T^{V_m}_{jj}=1$ for $ j\in V_m,$ and $T^{V_m}_{jj}=0$ otherwise. To simplify the notation, we have $\sum_{j\in V_m}C_{jj}=\text{tr}(T^{V_m}C),$  and $\sum_{j \in V_m} q_j=\text{tr}(T^{V_m}\tilde{\theta} C \tilde{\theta}).$ Because $C$ is positive definite,  $\sum_{j \in V_m} q_j>0.$ Therefore, the quadratic equation has one unique positive root.  Alternating the updating scheme throughout all the $\theta_{E_s},$ and $\theta_{V_m}$ until convergence, we obtain the penalized sparse estimate of the concentration matrix under RCON model.

\begin{center}
2.2  E{\scriptsize STIMATION OF RCOR MODEL}
\end{center}

An RCOR $(\mathcal{V},\mathcal{E})$ model with vertex colouring $\mathcal{V}$ and edge coloring $\mathcal{E}$ is obtained by restricting the elements of $\theta$ as follows: (a) All diagonal elements of $\theta$ (inverse partial variances) corresponding to vertices in the same vertex colour class must be identical. (b) All partial correlations corresponding to edges in the same edge colour class must be identical. Given an edge color class, for all edges $(i,j)\in E_s,$ $\rho_{ij}$ are all equal and hence denoted as $\rho_{E_s}.$ This also implies given a vertex color class, for all vertices $(i)\in V_m,$ $\theta_{ii}$ are all equal and hence denoted as $\theta_{V_m},$ and $\sigma^{ii}$ are all equal and hence denoted as $\sigma_{V_m},$
We formulate the composite likelihood in terms of 
$\rho_{E_s}$ and $\sigma_{V_m}.$  

For notational convenience, define a $p\times p$ matrix $\tilde{\rho}$ with $\tilde{\rho}_{ij}=\rho_{ij}$ for $i \neq j$ and $\tilde{\rho}_{ij}=0$ for $i=j.$ Let $\tilde{\rho}_j$ denote the $j$th column of the matrix $\tilde{\rho}.$ Define a $p$-element vector $\sigma_D=(\sigma^{11},\dots,\sigma^{pp})^T.$ The composite likelihood is formulated as $$
\ell_c(\rho,\sigma)=\frac 1 2 \sum_{j=1}^p\{ n \log \sigma^{jj}+ \frac 1 {\sigma^{jj}} ||X_j-X(\tilde{\rho}_j\odot \sigma_D^{-\frac 1 2})(\sigma^{jj})^{\frac 1 2}||_2^2 \}
.$$

We propose to estimate the sparse RCOR model by minimizing the following penalized composite loglikelihood $Q(\rho,\sigma)$:
$$
\min_{\theta_{E_s},1\leq s \leq l,\theta_{V_m},1\leq m \leq k}\ell_c(\rho,\sigma)+n\lambda\sum_{s}|\rho_{E_s}|
.$$

The partial derivative of $Q(\rho,\sigma)$ with respect to the partial correlation is as follows:
\begin{align}
\begin{split}
&\frac{\partial Q(\rho,\sigma)}{\partial \rho_{E_s}}\\
=&n\rho_{E_s}\text{tr}\biggl((\sigma^{-1/2}T^{E_s})^T C (\sigma^{-1/2}T^{E_s})\biggr)+
n\text{tr}\biggl((\sigma^{-1/2}\tilde{\rho}\odot T^{E_s^c})^T C (\sigma^{-1/2}T^{E_s})\biggr)\\
&-\text{tr}\biggl((X\sigma^{-1/2})^T X (\sigma^{-1/2}T^{E_s})\biggr)+n\text{sgn}(\theta_{E_s}).
\end{split}
\end{align}

The thresholded estimate of the partial correlation takes the following form:
$$
\tilde{\rho}_{E_s}=\frac{S(\text{tr}(T^{E_s} (\sigma^{-\frac 1 2} C \sigma^{-\frac 1 2}))-\text{tr}(T^{E_s}(T^{E_s^c}\odot \tilde{\rho}) (\sigma^{-\frac 1 2} C \sigma^{-\frac 1 2})),\lambda)}{\text{tr}(T^{E_s}.T^{E_s}  (\sigma^{-\frac 1 2} C \sigma^{-\frac 1 2})) }.
$$

The partial derivatives with respect to $\sigma_{V_m}$ is as follows:
\begin{align}
\begin{split}
&\frac{\partial \ell(\rho,\sigma)}{\partial \sigma_{V_m}}\\
=&\frac n 2 \{\frac{|V_m|}{\sigma_{V_m}} -\frac{\sum_{j \in V_m} x_j^T x_j}{n \sigma^2_{V_m}}                                   +\frac{\sum_{i \in V_m}\sum_{j \in V_m} 2 x_i^T x_j \tilde{\rho}_{ij}}{n \sigma^2_{V_m}}  +\frac 2 n \sigma_{V_m}^{-\frac 3 2} \sum_{(i,j);i \in V_m,j \notin V_m} x_i^T x_j  \tilde{\rho}_{ij}/\sqrt{\sigma_{jj}} \\
&-\frac 1 {n \sigma^2_{V_m}} \sum_{j=1}^p \sum_{i\in V_m}\sum_{i'\in V_m} x_i^T x_{i'}   \tilde{\rho}_{ij}\tilde{\rho}_{i'j}                         -\frac 1 {n \sigma^{\frac 3 2}_{V_m}} \sum_{j=1}^p \sum_{i\in V_m}\sum_{i'\notin V_m} x_i^T x_{i'}   \tilde{\rho}_{ij}\tilde{\rho}_{i'j}/\sqrt{\sigma^{i'i'}}\}.
\end{split}
\end{align}

Re-express the above expression in terms of $y=\sqrt{\sigma_{V_m}}.$   We solve the equation
$$
|V_m| y^2-b y-a=0,
$$
with 
\begin{align}
\begin{split}
a&=\sum_{j \in V_m} x_j^T x_j/n- \sum_{i \in V_m}\sum_{j \in V_m} 2x_j^T x_i \tilde{\rho}_{ij}/n+\sum_{j=1}^p \sum_{i \in V_m} \sum_{i' \in V_m} x_i^T x_{i'} \tilde{\rho}_{ij}\tilde{\rho}_{i'j}/n\\
&=\text{tr}(T^{V_m} C)-2\text{tr}(T^{V_m} C T^{V_m}\tilde{\rho}))+\text{tr}(\tilde{\rho}T^{V_m} C T^{V_m} \tilde{\rho})
\end{split}
\end{align}
and
\begin{align}
\begin{split}
b&= -\sum_{i \in V_m}\sum_{j \notin V_m} 2x_j^T x_i \tilde{\rho}_{ij}/(n\sqrt{\sigma^{jj}})+\sum_{j=1}^p \sum_{i \in V_m} \sum_{i' \notin V_m} x_i^T x_{i'} \tilde{\rho}_{ij}\tilde{\rho}_{i'j}/(n\sqrt{\sigma^{i'i'}})\\
&=-2\text{tr}(T^{V_m} C \sigma^{-1/2} T^{V_m^c} \tilde{\rho})+\text{tr}(\tilde{\rho}T^{v_m^c}\sigma^{-1/2}C T^{V_m}\tilde{\rho}).
\end{split}
\end{align}
The solution would be
$$
y=\frac{b+\sqrt{b^2+4a|V_m|}}{2|V_m|}.
$$
The positive solution is unique because 
\begin{align}
\begin{split}
a=&\text{tr}\biggl(C(T^{V_m}-\tilde{\rho}T^{V_m} )^T(T^{V_m}-\tilde{\rho}T^{V_m})\biggr)>0.
\end{split}
\end{align}

\begin{center}
2.3  A{\scriptsize SYMPTOTIC PROPERTIES}
\end{center}

In this section, we discuss the asymptotic properties of the penalized composite likelihood estimates for sparse symmetric Gaussian graphical models.  In terms of the choice of penalty function, there are many penalty functions available.  As the LASSO penalty, $p_{\lambda}(|\theta_l|)=\lambda|\theta_{l}|,$ increases
linearly with the size of its argument, it leads to biases for the
estimates of nonzero coefficients. To attenuate such estimation
biases, Fan and Li (2001) proposed the SCAD penalty. The penalty
function satisfies $p_{\lambda}(0)=0,$ and its first-order
derivative is
\begin{align*}
p_{\lambda}'(\theta)=\lambda\{I(\theta\leq \lambda)+
\frac{(a\lambda-\theta)_{+}}{(a-1)\lambda}I(\theta>
\lambda)\},\,\,\text{for}\,\,\theta\geq 0,
\end{align*}
where $a$ is some constant, usually set to $3.7$ (Fan and Li,
2001), and $(t)_{+}=t I(t>0)$ is the hinge loss function.
The SCAD penalty
function does not penalize as heavily as the $L_1$ penalty
function on parameters with large values. It has been shown that with probabability tending to one, the likelihood estimation with the SCAD penalty
not only selects the correct set of significant covariates, but
also produces parameter estimators as efficient as if we know the
true underlying sub-model (Fan \& Li, 2001). Namely, the estimators have the so-called ORACLE property. 
However, it has not been investigated if the oracle
property is also enjoyed by composite likelihood estimation of GGM with the SCAD
penalty. The following discussion is focused on the RCON model but it can be easily extended to RCOR model.

For notational convenience, let $z=\{E_s:\theta_{E_s}\neq 0\}\cup \mathcal{V}$ denote all the nonzero edge classes and all vertex classes and  $z^c=\{E_s:\theta_{E_s}=0\}$ denote all the zero edge classes. The parameter vector can be expressed as $\theta=(\theta_{E_1},\dots,\theta_{E_l},\theta_{V_1},\dots,\theta_{V_k}).$ Let $\theta_0$ denote the true null value.

\begin{thm}
\label{thm1}
Given the SCAD penalty function $p_{\lambda}(\theta),$ if $\lambda_n\rightarrow 0,$ and $\sqrt{n}\lambda_n \rightarrow
\infty$ as $n\rightarrow \infty,$ then there exist a local maximizer $\hat{\theta}$ to $Q(\theta)$ and $||\hat{\theta}-\theta_0||=O_p(n^{\frac 1 2}).$ Furthermore, we have
$$
\lim_{n\rightarrow \infty} P(\hat{\theta}_{z^c}=0)=1.
$$ 
\end{thm}

\begin{proof}
Consider a ball $||\theta-\theta_0||\leq Mn^{-\frac 1 2}$ for some finite $M.$ Applying Taylor
Expansion, we obtain:
\begin{align}
\begin{split}
\partial Q(\theta)/\partial \theta_{j}&=\partial \ell_c(\theta)/\partial \theta_{j}-n
p'_{\lambda_n}(|\theta_j|)\text{sign}(\theta_j)\\
&=\partial \ell_c(\theta_0)/\partial \theta_{j}+\sum_{j'\in (\mathcal{E}\cup\mathcal{V})} (\theta_{j'}-\theta_{j'0}) \partial^2 \ell_c(\theta^*)/\partial \theta_j \theta_{j'} -n
p'_{\lambda_n}(|\theta_j|)\text{sign}(\theta_j),
\end{split}
\end{align}
for $j\in (\mathcal{E}\cup\mathcal{V})$ and some $\theta^*$ between $\theta$ and $\theta_0.$
As $E(\partial \ell_c(\theta_0)/\partial \theta_{j})=0,$  $\partial \ell_c(\theta_0)/\partial \theta_{j}=O_p(n^{\frac 1 2}).$
As $|\theta^*-\theta|\leq M n^{-\frac 1 2}$ and  $\partial^2 \ell_c(\theta^*)/\partial \theta_j \theta_{j'}=O_p(n)$ componentwise.
First we consider $j\in z^c$.
Because
$
\text{lim inf}_{n\rightarrow \infty}\text{lim
inf}_{\beta\rightarrow 0+}p'_{\lambda_n}(\beta)/\lambda_n>0,
$
and $\lambda_n\rightarrow 0,$ and $\sqrt{n}\lambda_n \rightarrow
\infty$ as $n\rightarrow \infty,$ the third term dominates the the first two terms. Thus the sign of $\partial Q(\theta)/\partial \theta_{j}$ is completely determined by the
sign of $\beta_j.$ This entails that inside this
$Mn^{-1/2}$ neighborhood of $\beta_0$, $\partial Q(\theta)/\partial \theta_{j}>0,$ when $\theta_j<0$ and
$\partial Q(\theta)/\partial \theta_{j}<0,$ when $\theta_j>0.$ Therefore for any local maximizer $\hat{\theta}$ inside this ball, then $\hat{\theta}_{j}=0$  with probability tending to one.  As $p_{\lambda_n}(0)=0,$ we obtain
\begin{align}
\begin{split}
Q(\theta)-Q(\theta_0)&=\ell_c(\theta)-\ell_c(\theta_0)-n\sum_{j\in(\mathcal{E}\cup\mathcal{V})}\bigl(p_{\lambda_n}(|\theta_j|)-p_{\lambda_n}(|\theta_{j0}|)\bigr)\\
&\leq (\theta-\theta_0)^T \frac{\partial \ell_c(\theta_0)}{\partial \theta}+(\theta-\theta_0)^T \frac{\partial^2 \ell_c(\theta^*)}{\partial \theta^2}(\theta-\theta_0)\\
&-n\sum_{j\in z} \biggl(p'_{\lambda_n}(|\theta_{j0}|)\text{sign}(\theta_{j0})(\theta_j-\theta_{j0})+p^{''}_{\lambda_n}(|\theta_{j0}|)(\theta_j-\theta_{j0})^2(1+o(1))\biggr).
\end{split}
\end{align}
For $n$ large enough and $\theta_{j0}\neq 0,$
$p'_{\lambda}(|\theta_{j0}|)=0$ and
$p''_{\lambda}(|\theta_{j0}|)=0.$ 
Furthermore, $\partial^2 \ell_c(\theta^*)/ \partial \theta^2$ converges to $H(\theta)$ in probability, which is negative definite. Thus, we have $Q(\theta)\leq Q(\theta_0)$ with probability tending to one for $\theta$ on the unit ball. This implies there exists a local maximizer of $\hat{\theta}$ such that $|\hat{\theta}-\theta_0|=O_p(n^{-1/2}).$
\end{proof}

Next, we establish the asymptotic distribution of the estimator $\hat{\theta}$. Let $\theta_z$ denote the sub-vector of nonzero parameters in $\theta.$ 
Define a matrix
$\Sigma_1=\text{diag}\{p''_{|\lambda_n|}(\theta_{j0});j\in z\},
$
and a vector
$
b_1=(p'_{\lambda_n}(\theta_{j})\text{sign}(\theta_{j0});j\in z).
$
Let $H_{zz}$ denote the sub-matrix of $H(\theta)$ and $V_{zz}$ denote the sub-matrix of $V(\theta)$ corresponding to the subset of $z.$

\begin{thm}
\label{thm2}
Given the SCAD penalty function $p_{\lambda}(\theta),$
if $\lambda_n\rightarrow 0$ and $\sqrt{n}\lambda_n\rightarrow
\infty,$ as $n \rightarrow \infty,$ then the sub-vector of the root-n consistent estimator  $
\hat{\theta}_z$ has the following asymptotic
distribution:
$$
\sqrt{n}(H_{zz}+\Sigma_1)\{\hat{\theta}_z-\theta_{z0}+(H_{zz}+\Sigma_1)^{-1}b_1\}\rightarrow
N\{0,V_{zz}\}, \, \text{as}\, n\rightarrow \infty.
$$
\end{thm}

\begin{proof}
Based on Taylor expansion presented in Proof to Theorem 1, we have
\begin{align}
0=\frac{\partial Q(\hat{\theta})}{\partial \theta_z}=\frac{\partial \ell_c(\theta_0)}{\partial \theta_z}+\frac{\partial^2 \ell_c(\theta^*)}{\partial \theta_z \partial \theta_z^T}(\hat{\theta}_z-\theta_{z0})-nb_1-n(\Sigma_1+o(1))(\hat{\theta}_z-\theta_{z0}).
\end{align}
As $\hat{\theta}\rightarrow \theta_0$ in probability, $\frac 1 n \frac{-\partial^2 \ell_c(\theta^*)}{\partial \theta_z \partial \theta_z^T}\rightarrow H_{zz}$ in probability.
The limiting distribution of $
\frac {1} {\sqrt{n}} \frac{\partial \ell_c(\theta_0)}{\partial \theta_z}$ is $N\{0,V_{zz}\}.$
According to Slutsky's theorem, we
have$
\sqrt{n}(H_{zz}+\Sigma_1)\{\hat{\theta_z}-\theta_{z0}+(H_{zz}+\Sigma_1)^{-1}b_1\}\rightarrow
N\{0,V_{zz})\}.$
\end{proof}

Next we discuss the estimation of the Hessian matrix $H_{zz}$ and the variability matrix $V_{zz}$. As the second differentiation is easy to calculate, we obtain
$
\hat{H}_{zz}=\partial^2 \ell_c(\theta)/\partial \theta_z \partial \theta_z^T |_{\hat{\theta}}.
$
The variability matrix based on sample covariance matrix of the composite score vectors is computationally harder as we need to compute the composite score vector for each observation, where the number of observations can be large. Alternatively, we perform bootstrap to obtain the
$$
\hat{V}_{zz}=\frac{1}{n(m-1)}\sum_{l=1}^m (S^{(m)}(\hat{\theta})-\overline{S})^T (S^{(m)}(\hat{\theta})-\overline{S}),
$$
where $S(\theta)=\partial \ell_c(\theta)/\partial \theta_z,$
$S^{(m)}(\hat{\theta})$ denotes the score vector evaluated with the composite estimator obtained from the original sample and  the data from the $m$th bootstrap sample and $\overline{S}=\sum_{l=1}^m S^{(m)}(\hat{\theta})/m.$  In pratice, we only need a moderate number of bootstrap samples to obtain $\hat{V}_{zz}.$

\begin{center}
3. N{\scriptsize UMERICAL ANALYSIS}
\end{center}

We analyze the ``math" data set from Mardia et al. (1979), which consists of 88 students in 5 different mathematics subjects: Mechanics (me), Vectors (ve), Algebra (al), Analysis (an) and Statistics (st). The model with symmetry proposed by Hojsgaard and Lauritzen (2008) has vertex color classes \{al\}, \{me, st\}, \{ve, an\} and edge color classes \{(al,an)\}, \{(an,st)\}, \{(me,ve), (me,al)\}, and \{(ve,al), (al,st)\}. We perform composite likelihood estimation on this symmetric model with no penalty imposed on the parameters. In Table 1, the composite likelihood estimates and their standard deviations calculated through bootstraps are compared with those obtained by maximum likelihood estimator and a naive estimator. The naive estimator estimates the edge class parameters and vertex class parameters by simply averaging all the values belonging to the same class in the inverse sample covariance matrix. All three methods yield results that are very close to each other. 
 
Next we examine the performance of the unpenalized composite likelihood estimator on large matrices. First we consider the RCON model. We simulate under different scenarios with $n$ varying from 250 to 1000 and $p$ varying from 40, 60 to 100.  We include 30 different edge classes and 20 different vertex classes. We simulate a sparse matrix with $\theta_{\mathcal{E}}=(0_{25},0.2591, 0.1628,-0.1934,0.0980,0.0518),$ and $\theta_{\mathcal{V}}=(1.3180, 1.8676, 1.788004,$ $ 1.7626, 1.6550,$ $1.1538, 1.3975, 1.7877,
 1.7090, 1.6931,$ $ 1.46313, 1.5131, 1.7084, $ $1.7344, 1.1441,$ $ 1.8059,
 1.7446,$ $ 1.8522,$ $ 1.3146,$ $ 1.1001),$ where $0_p$ denotes a zero vector of length $p.$ The number of nonzero edges ranges from about 250 to 1640. In Table 2, we compare the sum of squared errors of the composite likelihood estimates with the naive estimates from 100 simulated data sets. The proposed composite likelihood estimates consistenly enjoy much smaller sum of squared errors across all settings. 

We also investigate the empirical performance of the proposed composite likelihood estimator under the RCOR model. We simulate under different scenarios with $n$ varying from 250 to 1000 and $p$ varying from 40, 60 to 100.  We include 30 different edge classes and 20 different vertex classes. We simulate a sparse matrix with $\rho_{\mathcal{E}}=(0_{26},0.1628,$ $-0.1534,$ $0.0980,0.0518)$ and $\theta_{\mathcal{V}}=(3.0740, 3.6966, 3.7772, 3.5475, $ $3.2841, 3.4699, 3.7235,$ $ 3.5987,$ $
3.3313,$ $ 3.8183,$ $ 3.9236,$ $ 3.9008, $ $3.9011, 3.0470,$ $ 3.0139,$ $ 3.2072,
$ $3.8438, 3.4823,$ $ 3.9373, 3.0125.)$ In table 3, we provide the errors $||\hat{\rho}_{\mathcal{E}}-\rho_{\mathcal{E}}||_2$ and $||\sqrt{\hat{\sigma}}_{\mathcal{V}}-\sqrt{\hat{\sigma}}_{\mathcal{V}}||_2$ for the composite likelihood estimates and the naive estimates from 100 simulated data sets. With regard to the estimated partial correlations, the composite likelihood estimates yield consistently smaller errors compared to the naive estimates. With regard to the conditional standard deviations,  the composite likelihood estimates yield slightly larger errors under sample size $n=250,$ and $n=500.$ With sample size $n=1000,$ the composite likelihood estimates have smaller errors than the naive estimates. For example, with $p=100$ and the number of true edges close to 1300, the naive estimate for the conditional standard deviation has error $1.8116,$ while the composite likelihood estimate has error $0.2923.$

We further examine the empirical performance of the penalized composite likelihood estimator. We simulate the RCON model using the same settings as of Table 1. We consider different scenarios with $n=250$ or $n=500,$ and $p=40$ or $p=60.$ We use the penalized composite likelihood estimator to estimate the sparse matrix. The tuning parameter is selected by composite BIC, which is similar to BIC with the first term replaced by the composite likelihood evaluated at the penalized composite likelihood estimates.  For each setting, 100 simulated data sets are generated and for each data we calculate the number of false negatives  and false positives.  In Table 4, it is shown that the proposed method has satisfactory model selection property with very low false negative and false positive results. For example, with $n=500$ and $p=60$, each simulated data set has an average number of $1474$ zero edges and $325$ nonzero edges. The proposed method identifies an average of zero false negative result and 0.58 false positive result. The size of the tuning parameters is also listed in Table 4. 

\begin{center}
5. A{\scriptsize PPLICATTION}
\end{center}
We apply the proposed method on a real biological data set.  The experiment was conducted to examine how GM-CSF modulates global changes in neutrophil gene expressions (Kobayashi et al, 2005). Time course summary PMNs were isolated from venous blood of healthy individuals. 
Human PMNs (107) were cultured with and without 100 ng/ml GM-CSF  for up to 24 h.  The Experiment was performed in triplicate, using PMNs from three healthy individuals for each treatment. There are in total 12625 genes monitored, each gene is measured for 9 replications at time 0, and measured for 6 times at time 3, 6, 12, 18, 24h. At each of these 5 points, 3 measurements were obtained for treatment group and 3 measurements were obtained for control group. We first proceed with standard gene expression analysis. For each gene, we perform an ANOVA test on the treatment effect while aknowledging the time effet. We rank the F statistic for each gene and select the top 200 genes who have the most significant changes in expression between treatment and control group. Our goal is to study the networks of these 200 genes and also compare the network of the 200 genes between the treatment and control. 
We perform clustering analysis on the selected 200 genes, where the genes clustered together can be viewed as a group of genes who share similar expression profiles. This imposes symmetry constraints to the graphical modelling. We cluster these top 200 genes into 10 clusters based on K-means method. Therefore, there are in total of 55 edge classes and 10 vertex classes to be estimated based on a 200 by 200 data matrices. We perform penalized estimation and compare the result of the estimated edges between the treatment versus control. The estimated between-cluster edges are provided in Figure 1.  It is observed that although most between-cluster interactions are small, there are a few edges with large values indicating strong interactions. It is also observed that the edge values obtained from the treatment group and the control group are mostly comparable and only a few edges exhibit big differences. For instance, edges between cluster 1 and 5 and between cluster 4 and 6 have big differences in treatment group versus control group. These findings are worth further biological investigation to unveil the physical mechanism underlying the networks.  

\begin{center}
6. C{\scriptsize ONCLUSION}
\end{center}

When there are both sparsity and symmetry constrainsts on the graphical model, the penalized composite likelihood formulation based on conditional distributions offers an alternative way to perform the estimation and model selection. The estimation avoids the inversion of large matrices. It is shown that the proposed penalized composite likelihood estimator will threshold the estimate for zero parameters to zero with probability tending to one and the asymptotic distribution of the estimates for non-zero parameters follow the multivariate normal distribution as if we know the true submodel containing only non-zero parameters.

\medskip
\begin{center}
A{\scriptsize CKNOWLEDGEMENT}
\end{center}
\noindent
This research is supported by NSERC grants held by Gao and Massam.

\vskip 1in
\noindent{\large\bf References}
\begin{description}
%\item BEEIMAN, L. (1996). Heuristics of Instability and
%Stabilization in Model Selection. \emph{Ann. Statist.}
%\textbf{24}, 2350-83.

\item Banerjee, O., Ghaoui, L. E., and D'Aspremont, A. (2007).
Model selection through sparse maximum likelihood estimation.
\emph{Journal of Machine Learning Research} \textbf{9}, 485-516.

\item Dempster, A. P. (1972). Covariance selection.
\emph{Biometrika} \textbf{32}, 95-108.

\item Edwards, D. M. (2000). \emph{Introduction to Graphical
Modelling.} Springer, New York.

\item Fan, J. and Li, R. (2001). Variable selection via nonconcave
penalized likelihood and its oracle properties. \emph{Journal of
the American  Statistical Association} \textbf{96}, 1348-60.

\item Fan, J., Feng, Y., and Wu, Y. (2009). Network exploration via
the adaptive LASSO and SCAD penalties. \emph{The Annals of Applied
Statistics} \textbf{3}, 521-541.

\item Friedman, J., Hastie, T., and  Tibshirani, R. (2008). Sparse
inverse covariance estimation with the graphical lasso.
\emph{Biostatistics} \textbf{9}, 432-441.

\item Højsgaard, S and Lauritzen, S. L. (2008) Graphical Gaussian Models with Edge and Vertex symmetries. \emph{Journal of Royal Statistical Society, Series B} \textbf{70}, 1005-1027.

\item Kobayashi, S. D, Voyich, J. M., Whitney, A. R., and DeLeo, F. R. (2005) Spontaneous neutrophil apoptosis and regulation of cell survival by granulocyte macrophage-colony stimulating factor. \emph{Journal of  Leukocite  Biology} \textbf{78}, 1408-18.

\item Lam, C. and  Fan, J. (2009). Sparsistency and rates of
convergence in large covariance matrix estimation. \emph{The
Annals of Statistics} \textbf{37}, 4254-4278.

\item Lauritzen, S. L. (1996). \emph{Graphical Models}. 
Clarendon Press, Oxford.

\item Meinshausen, N. and B\"{u}hlmann, P. (2006).
High-dimensional graphs with the Lasso. \emph{Annals of Statistics}
\textbf{34}, 1436-62.

\item Whittaker, J. (1990) \emph{Graphical Models in Applied
Multivariate Statistics}. Wiley, Chichester.

\item Yuan, M. and Lin, Y. (2007). Model selection and estimation
in the Gaussian graphical model. \emph{Biometrika} \textbf{94},
19-35.

\item Zou, H. and Li, R. (2008) One-step sparse estimates in
nonconcave penalized likelihood models (with discussion).
\emph{Annals of Statistics}  \textbf{36}, 1509-1533.

\item Cox, D. R. and Reid,  N. (2004).  A note on
pseudolikelihood constructed from marginal densities.
\emph{Biometrika} {\bf 91}, 729-737.

\item Dempster, A. P., Laird, N. M.  and Rubin, D. B. (1977).
 Maximum likelihood from incomplete data via the EM algorithm,
\emph{Journal of the Royal Statistical Society} B {\bf 39},
1-38.

\item Varin, C. (2008).  On composite marginal likelihoods. {\em
AStA Advances in Statistical Analysis}, \textbf{92}, 1-28.

\end{description}

\begin{center}
 A{\scriptsize PPENDIX}
\end{center}
\begin{itemize}
\item  The detailed derivation of the first derivatives with respect to $\theta_{E_s}$ under RCON model is as follows:
\begin{align}
\begin{split}
&\frac{\partial Q(\theta)}{\partial \theta_{E_s}} \\
=&\sum_{j=1}^p \frac{1}{2\sigma^{jj}} \frac{\partial ||X_j+X\tilde{\theta}_j\sigma^{jj}||_2^2}{\partial \theta_{E_s}}+n\text{sgn}(\theta_{E_s})\\
=&\sum_{j=1}^p \frac 1 {\sigma^{jj}} (X_j+X\tilde{\theta}_j\sigma^{jj})^T \frac{\partial (X_j+X\tilde{\theta}_j\sigma^{jj})}{\partial \theta_{E_s}}+n\text{sgn}(\theta_{E_s})\\
=&\sum_{j=1}^p (X_j+X\tilde{\theta}_j\sigma^{jj})^T (\sum_{i;(i,j)\in E_s} X_i)+n\text{sgn}(\theta_{E_s})\\
=&\sum_{j=1}^p( X_j^T(\sum_{i;(i,j)\in E_s} X_i)+\sigma^{jj}(\sum_{i;(i,j)\in E_s}X_i^T(\sum_{l;(l,j)\in E_s}X_l \theta_{E_s}+\sum_{l;(l,j)\notin E_s}X_l \theta_{lj})))+n\text{sgn}(\theta_{E_s})\\
=&( \sum_{j=1}^p\sum_{i;(i,j)\in E_s}\sum_{l;(l,j)\in E_s}\sigma^{jj}X_i^T X_l)\theta_{E_s}+ \biggl(\sum_{j=1}^p X_j^T(\sum_{i;(i,j)\in E_s} X_i)+\\
&\sigma^{jj} \sum_{i;(i,j)\in E_s}\sum_{l;(l,j)\notin E_s}X_i^T X_l \theta_{lj}\biggr)+n\text{sgn}(\theta_{E_s}).
\end{split}
\end{align}

 \item The detailed derivation of the first derivatives with respect to $\theta_{V_m}$ under RCON model is as follows:
\begin{align}
\begin{split}
&\frac{\partial Q(\theta)}{\partial \sigma_{V_m}} \\
=& \frac 1 2 \sum_{j \in V_m} \frac n {\sigma^{jj}}+\partial \{\frac{(X_j+X\tilde{\theta}_j \sigma^{jj})^T (X_j+X\tilde{\theta}_j \sigma^{jj})}{\sigma^{jj}}\}/\partial \sigma^{jj}\\
=& \frac n 2 \sum_{j \in V_m}( \frac 1 {\sigma^{jj}}-\frac{C_{jj}}{(\sigma^{jj})^2}+q_j),
\end{split}
\end{align}
where $C_{ij}=x_i^T x_j/n,$ and $q_j=\sum_{l=1}^p \sum_{l'=1}^p C_{ll'} \tilde{\theta}_{lj}\tilde{\theta}_{l'j}.$

\item The detailed derivation of the first derivatives with respect to $\rho_{E_s}$ under RCOR model is as follows:
\begin{align}
\begin{split}
&\frac{\partial Q(\rho,\sigma)}{\partial \rho_{E_s}}\\
=& \sum_{j=1}^p \frac 1 {\sqrt{\sigma^{jj}}}\biggl(X(\tilde{\rho}_{j}\odot \sigma_D^{-1/2})\sqrt{\sigma^{jj}}-X_j\biggr)^T X (\frac{\partial \tilde{\rho}_j}{\partial \rho_{E_s}}\odot \sigma_D^{-1/2})+n\text{sgn}(\theta_{E_s}).
\end{split}
\end{align}
Note that $(\tilde{\rho}_{j}\odot \sigma_D^{-1/2}) =(\sigma^{-1/2}\tilde{\rho})_{[,j]},$ the $j$th column of the matrix. 
Also we have the vector $\frac{\partial \tilde{\rho}_j}{\partial \rho_{E_s}}\odot \sigma_D^{-1/2}=(\sigma^{-1/2}T^{E_s})_{[,j]}$ the $j$th column of the matrix. Furthermore, $\tilde{\rho}=\sum_{s'}\rho_{E_{s'}}T^{E_{s'}}.$ This leads to:
\begin{align}
\begin{split}
&\frac{\partial Q(\rho,\sigma)}{\partial \rho_{E_s}}\\
=& \sum_{j=1}^p \rho_{E_s}(\sigma^{-1/2}T^{E_s})_{[,j]}X^T X (\sigma^{-1/2}T^{E_s})_{[,j]}+
(\sigma^{-1/2}\tilde{\rho}\odot T^{E_s^c})_{[,j]}X^T X (\sigma^{-1/2}T^{E_s})_{[,j]}\\
&-\frac 1 {\sqrt{\sigma^{jj}}}X_j^T X(\sigma^{-1/2}T^{E_s})_{[,j]}+n\text{sgn}(\theta_{E_s})\\
=&n\rho_{E_s}\text{tr}\biggl((\sigma^{-1/2}T^{E_s})^T C (\sigma^{-1/2}T^{E_s})\biggr)+
n\text{tr}\biggl((\sigma^{-1/2}\tilde{\rho}\odot T^{E_s^c})^T C (\sigma^{-1/2}T^{E_s})\biggr)\\
&-\text{tr}\biggl((X\sigma^{-1/2})^T X (\sigma^{-1/2}T^{E_s})\biggr)+n\text{sgn}(\theta_{E_s}).
\end{split}
\end{align}
\end{itemize}

\newpage
\begin{table}
 \caption{ Comparison of likelihood, composite likelihood, moment estimates on "math" dataset }
\label{tab1}

\begin{center}
{\scriptsize
\begin{tabular}{ccccccc}
\hline\hline parameter  & est &std & est & std &est &std\\\hline
  & likelihood & & composite & &moment &\\\hline
\hline
vcc1  & 0.0281  & 0.0037  &  0.0068 & 0.0005 & 0.0057 & 0.0005\\
vcc2  & 0.0059  &0.0006   & 0.0074  &0.0006  & 0.0098  &0.0013\\
vcc3   &0.0100   &0.0009   &0.0176  &0.0020   &0.0182  &0.0029\\
ecc1  &-0.0080  &0.0015  &-0.0062  &0.0009  &-0.0068  &0.0019\\
ecc2  &-0.0018   &0.0007  &-0.0008  &0.0005  &-0.0021  &0.0008\\
ecc3  &-0.0030   &0.0004  &-0.0027  &0.0002  &-0.0019  &0.0006\\
ecc4  &-0.0047   &0.0008  &-0.0051  &0.0005  &-0.0055  &0.0012\\
\hline\hline
\end{tabular}%
}
\end{center}
\end{table}

\begin{table}
 \caption{ Comparison of  $||\theta-\hat{\theta}||_2^2$ from composite likelihood and moment estimates on simulated large dataset for RCON model}
\label{tab1}

\begin{center}
{\scriptsize
\begin{tabular}{ccccc}
\hline\hline n  & p &comp & moment & \#true edges \\\hline
 
250 &  40& 0.2002 & 2.3671 & 256.7475\\ 
 & & (0.0757) &  (0.4580)  &  (15.5644) \\ 
250 &  60& 0.1109 & 5.6270 & 590.4040\\ 
 & & (0.0367) &  (0.7201)  &  (23.5606) \\ 
250 &  100& 0.0509 & 23.7040 & 1647.0707\\ 
 & & (0.0155) &  (2.0364)  &  (34.7461) \\ 
500 &  40& 0.0901 & 0.5482 & 256.7475\\ 
 & & (0.0272) &  (0.1439)  &  (15.5644) \\ 
500 &  60& 0.0588 & 1.0924 & 590.4040\\ 
 & & (0.0177) &  (0.1728)  &  (23.5606) \\ 
500 &  100& 0.0252 & 3.3530 & 1647.0707\\ 
 & & (0.0098) &  (0.2781)  &  (34.7461) \\ 
1000 &  40& 0.0467 & 0.1548 & 256.7475\\ 
 & & (0.0160) &  (0.0444)  &  (15.5644) \\ 
1000 &  60& 0.0282 & 0.2596 & 590.4040\\ 
 & & (0.0090) &  (0.0491)  &  (23.5606) \\ 
1000 &  100& 0.0125 & 0.6686 & 1647.0707\\ 
 & & (0.0037) &  (0.0684)  &  (34.7461) \\

\hline\hline
\end{tabular}%
}
\end{center}
\end{table}

\newpage
\begin{table}
 \caption{ Comparison of  the composite likelihood and moment estimates on simulated large dataset for RCOR model }
\label{tab2}

\begin{center}
{\scriptsize
\begin{tabular}{ccccccc}
\hline\hline n  & p &comp & moment & comp & moment &\#true edges \\
   &  &$||\hat{\rho}-\rho_{0})||_2$ & $||\tilde{\rho}-\rho_{0})||_2$ &$||\hat{\sigma}^{1/2}-\sigma_{0}^{1/2}||_2$ & $||\tilde{\sigma}^{1/2}-\sigma_{0}^{1/2}||_2$ & \\
\hline
250 &  40& 0.0317 & 0.0350 & 2.3869 & 2.2941 & 206.3200\\ 
 & & (0.0043) & ( 0.0050)  &  (0.0185) & (0.0179) & (13.5011)\\
250 &  60& 0.0196 & 0.0231 & 2.3886 & 2.2447 & 474.0400\\
 & & (0.0023) & ( 0.0029)  &  (0.0146) & (0.0149) & (22.0247) \\
250 &  100& 0.0097 & 0.0140 & 2.3905 & 2.1449 & 1316.9200\\ 
 & & (0.0015) & ( 0.0019)  &  (0.0118) & (0.0126) & (33.0795) \\
500 &  40& 0.0317 & 0.0350 & 0.9881 & 0.9226 & 206.3200\\ 
 & & (0.0043) & ( 0.0050)  &  (0.0131) & (0.0126) & (13.5011)\\
500 &  60& 0.0196 & 0.0231 & 0.9891 & 0.8874 & 474.0400\\ 
 & & (0.0023) & ( 0.0029)  &  (0.0103) & (0.0106) & (22.0247)\\ 
500 &  100& 0.0097 & 0.0140 & 0.9903 & 0.8167 & 1316.9200\\ 
 & & (0.0015) & ( 0.0019)  &  (0.0083) & (0.0089) & (33.0795)\\ 
1000 &  40& 0.0317 & 0.0350 & 0.0375 & 0.0615 & 206.3200\\ 
 & & (0.0043) & ( 0.0050)  &  (0.0062) & (0.0076) & (13.5011)\\ 
1000 &  60& 0.0196 & 0.0231 & 0.0301 & 0.0794 & 474.0400\\
 & & (0.0023) & ( 0.0029)  &  (0.0046) & (0.0071) & (22.0247)\\ 
1000 &  100& 0.0097 & 0.0140 & 0.0221 & 0.1255 & 1316.9200\\ 
 & & (0.0015) & ( 0.0019)  &  (0.0034) & (0.0063) & (33.0795)\\ 
\hline\hline
\end{tabular}%
}
\end{center}
\end{table}

\newpage
\begin{table}
\caption{Model selection performance of penalized composite likelihood based on 100 simulated datasets under each setting}
\label{tab1}

\begin{center}
{\scriptsize
\begin{tabular}{ccccccc}
\hline\hline n  & p &\#zero edge &\#true edges & fn &fp & tuning parameter \\\hline
 250 & 40 & 651.6300 &  120.5500& 27.8200 & 0.0000 & 1.2770\\ 
 &  & 7.7429  &12.9008  & (10.2152) &  (0.0000)  &  (0.3194) \\ 
250 & 60 & 1469.2300 &  323.0300& 2.3000 & 5.4400 & 1.4985\\ 
 &  & 19.5349  &16.4890  & (11.4111) &  (19.2518)  &  (0.2514) \\ 
500 & 40 & 651.6300 &  121.4700& 26.9000 & 0.0000 & 1.2650\\ 
 &  & 7.7429  &13.2432  & (10.6520) &  (0.0000)  &  (0.3705) \\ 
500 & 60 & 1474.0900 &  325.3300& 0.0000 & 0.5800 & 1.0910\\ 
 &  & 13.6929  &11.7903  & (0.0000) &  (5.8000)  &  (0.1961) \\ 
\hline\hline
\end{tabular}%
}
\end{center}
\end{table}

\begin{figure}[t]
\caption{Estimated between-cluster edges for treatment and control groups}
\vskip .1in
\begin{center}
\begin{tabular}{l}
\includegraphics[height=6in]{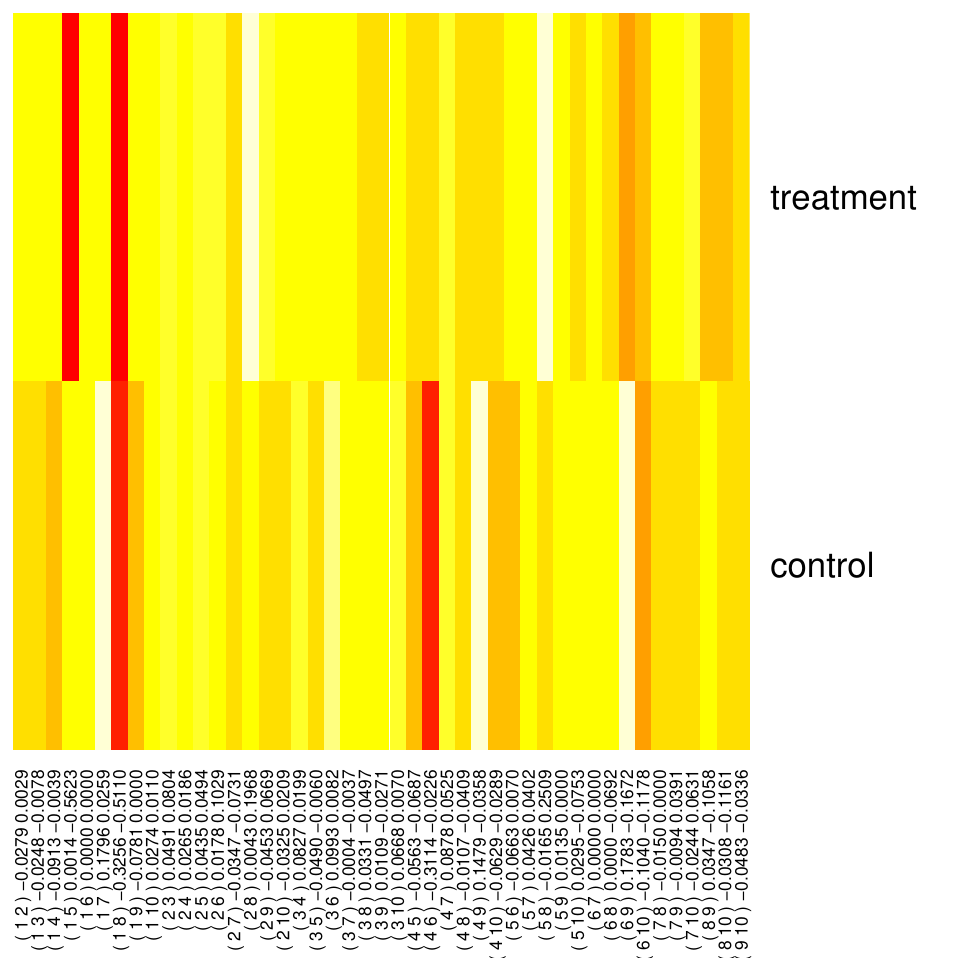}\\
{\small  Numbers in parenthesis indicate the cluster IDs, followed by the estimated $\hat{\theta}_{E_s}$}\\
{\small for the control and treatment groups. }
\end{tabular}
\end{center}
\end{figure}

\end{document}